\input ./style/arxiv-vmsta.cfg
\documentclass[numbers,compress,v1.0.1]{vmsta}
\usepackage{vtexbibtags}

\usepackage[OT2,T1]{fontenc}
\newcommand\textcyr[1]{{\fontencoding{OT2}\fontfamily{wncyr}\selectfont #1}}
\usepackage{authorquery,tmultirow}

\volume{5}
\issue{2}
\pubyear{2018}
\firstpage{225}
\lastpage{245}
\aid{VMSTA105}
\doi{10.15559/18-VMSTA105}
\articletype{research-article}

\DeclareMathOperator{\pr}{\mathsf{P}}
\DeclareMathOperator{\M}{\mathsf{E}}
\DeclareMathOperator{\D}{Var}

\DeclareMathOperator{\diag}{diag}

\DeclareMathOperator{\argmax}{argmax}
\startlocaldefs
\newcommand{\rrVert}{\Vert}
\newcommand{\llVert}{\Vert}

\allowdisplaybreaks

\newcommand{\bGamma}{{\boldsymbol\varGamma}}
\newcommand{\bSigma}{{\boldsymbol\varSigma}}
\newcommand{\ba}{{\mathbf a}}
\newcommand{\bA}{{\mathbf A}}
\newcommand{\bb}{{\mathbf b}}
\newcommand{\bp}{{\mathbf p}}
\newcommand{\bx}{{\mathbf x}}
\newcommand{\EE}{{\mathbb E}}
\newcommand{\bI}{{\mathbf I}}
\newcommand{\bX}{{\mathbf X}}
\newcommand{\bL}{{\mathbf L}}
\newcommand{\bD}{{\mathbf D}}
\newcommand{\bV}{{\mathbf V}}
\newcommand{\bu}{{\mathbf u}}
\newcommand{\bM}{{\mathbf M}}
\newcommand{\bY}{{\mathbf Y}}

\newcommand{\R}{{\mathbb R}}

\newcommand{\bydef}{\stackrel{\text{\rm def}}{=}}
\newcommand{\inprob}{\stackrel{\text{\rm P}}{\longrightarrow}}
\newcommand{\weak}{\stackrel{\text{\rm W}}{\longrightarrow}}
\newtheorem{thm}{Theorem}

\theoremstyle{definition}
\newtheorem{experiment}{Experiment}
\newtheorem*{note}{Note}
\endlocaldefs

\begin{aqf}
\querytext{Q1}{Note: it is assumed in what follows that '(RR)' are
assumptions 'in plural'.}
\querytext{Q2}{Seems that "is taken" fits better than holds.}
\querytext{Q3}{Do "are exponents from ... quadratic forms" means "are
exponentials of ... quadratic forms"?}
\querytext{Q4}{Check if the following formula is complete.}
\end{aqf}
\begin{document}

\begin{frontmatter}
\pretitle{Research Article}

\title{Confidence ellipsoids for regression coefficients by~observations from a mixture}

\author{\inits{V.}\fnms{Vitalii}~\snm{Miroshnichenko}\ead[label=e1]{vitaliy.miroshnychenko@gmail.com}}
\author{\inits{R.}\fnms{Rostyslav}~\snm{Maiboroda}\thanksref{cor1}\ead[label=e2]{mre@univ.kiev.ua}}
\thankstext[type=corresp,id=cor1]{Corresponding author.}

\address{\institution{Taras Shevchenko National University of Kyiv},
Kyiv, \cny{Ukraine}}



\markboth{V. Miroshnichenko and R. Maiboroda}{Confidence ellipsoids for
regression coefficients by
observations from a mixture}

\begin{abstract}
Confidence ellipsoids for linear regression coefficients are
constructed by observations from a mixture with varying concentrations.
Two approaches are discussed. The first one is the nonparametric
approach based on the weighted least squares technique. The second one
is an approximate maximum likelihood estimation with application of the
EM-algorithm for the estimates calculation.
\end{abstract}

\received{\sday{29} \smonth{1} \syear{2018}}
\revised{\sday{16} \smonth{5} \syear{2018}}
\accepted{\sday{19} \smonth{5} \syear{2018}}
\publishedonline{\sday{4} \smonth{6} \syear{2018}}
\end{frontmatter}
\setcounter{footnote}{0}

\section{Introduction}\label{SectIntr}

This paper is devoted to the technique of construction of confidence
ellipsoids for coefficients of linear regression in the case, when
statistical data are derived from a mixture with finite number of
components and the mixing probabilities (the concentrations of
components) are different for different observations.
These mixing probabilities are assumed to be known, but the
distributions of the components are unknown. (Such mixture models are
also known as mixtures with varying concentrations, see \cite
{Autin:TestDensities} and \cite{Maiboroda2003}).

The problem of estimation of regression coefficients by such mixed
observations was considered in the parametric setting in \cite{Faria}
and \cite{GrunLeisch}. The authors of these papers assume that the
distributions of regression errors and regressors are known up to some
unknown parameters.
The models considered in these papers are called
\textit{finite mixtures of regression models}. Some versions of
maximum likelihood estimates are used for the estimation of unknown
parameters of distributions and regression coefficients under these
models. The EM-algorithm is used for computation of the estimates.
(This algorithm is also implemented in \texttt{R} package \texttt
{mixtools}, \cite{Benaglia}. See \cite{McLachlan} for the general
theory of EM-algorithm and its application to mixture models).

In \cite{Liubashenko} a nonparametric approach was proposed under which
no parametric models on the distributions of components are assumed. A
weighted least squares technique is used to derive estimates for
regression coefficients. Consistency and asymptotic normality of the
estimates are demonstrated.

Note that in \cite
{Liubashenko,Maiboroda:StatisticsDNA,Maiboroda:AdaptMVC} a
nonparametric approach to the analysis of
mixtures with varying concentrations was developed in the case when the
concentrations of components (mixing probabilities) are known. Some
examples of real life data analysis under this assumption were
considered in \cite{Maiboroda:StatisticsDNA,Maiboroda:AdaptMVC}.

Namely, in \cite{Maiboroda:AdaptMVC} an application to the analysis of
the Ukrainian parliamentary elections-2006 was considered. Here the
observed subjects were respondents of the Four Wave World Values survey
(conduced in Ukraine in 2006) and the mixture components were the
populations of different electoral behavior adherents. The mixing
probabilities were obtained from the official results of voting in 27
regions of Ukraine.

In \cite{Maiboroda:StatisticsDNA} an application to DNA microarray data
analysis was presented. Here the subjects were nearly 3000 of genes of
the human genome. They were divided into two components by the
difference of their expression in two types of malignantly transformed
tissues. The concentrations were defined as a posteriori probabilities
for the genes to belong to a given component, calculated by
observations on the genes expression in sample tissues.

In this paper we will show how to construct confidence sets
(ellipsoids) for regression coefficients under both parametric and
nonparametric approaches. Quality of obtained ellipsoids is compared
via simulations.

The rest of the paper in organized as follows. In Section \ref
{SecModel} we present a formal description of the model. Nonparametric
and parametric estimates of regression coefficients and their
asymptotic properties are discussed in Sections \ref{SectNonp} and \ref
{SectParam}. Estimation of asymptotic covariances of these estimates is
considered in Section \ref{SecEstN}. The confidence ellipsoids are
constructed in Section \ref{SecCE}. Results of simulations are
presented in Section \ref{SecSimul}.
In Section \ref{SecApplication} we present a toy example of application
to a real life sociological data.
Section \ref{SecConcl} contains concluding remarks.

\section{The model}\label{SecModel}

We consider the model of mixture with varying concentrations. It means
that each observed subject $O$ belongs to one of $M$ different
subpopulations (mixture components). The number of component which the
subject belongs to is denoted by $\kappa(O)\in\{1,2,\dots,M\}$. This
characteristic of the subject is not observed. The vector of observed
variables of $O$ will be denoted by $\xi(O)$. It is considered as a
random vector with the distribution depending on the subpopulation
which $O$ belongs to. A~structural linear regression model will be used
to describe these distributions. (See \cite{Seber} for general theory
of linear regression).

That is, we consider one variable $Y=Y(O)$ in $\xi(O)=(Y(O),X^1(O),\dots
,\allowbreak X^d(O))^T$ as a response and all other ones $\bX
(O)=(X^1(O),\dots,X^d(O))^T$ as regressors in the model
\begin{equation}
\label{EqModel} Y(O)=\sum_{i=1}^d
b_i^{\kappa(O)}X^i(O)+\varepsilon(O),
\end{equation}
where $b_i^k$, $i=1,\dots,d$, $k=1,\dots,M$ are unknown regression
coefficients for the $k$-th component of the mixture, $\varepsilon(O)$
is the error term. Denote by $\bb^k=(b_1^k,\dots,b_d^k)^T$ the vector
of the $k$-th component's coefficients. We consider $\varepsilon(O)$ as
a random variable and assume that
\[
\M\bigl[\varepsilon(O)\ |\ \kappa(O)=m\bigr]=0, \quad m=1,\dots,M,
\]
and
\[
\sigma^2_m=\D\bigl[\varepsilon(O)\ |\ \kappa(O)=m\bigr]<
\infty.
\]
($\sigma^2_m$ are unknown).

It is also assumed that the regression error term
$\varepsilon(O)$ and regressors $\bX(O)$ are conditionally independent
for fixed $\kappa(O)=m$, $m=1,\dots,M$.

The observed sample $\varXi_n$ consists of values $\xi_j=(Y_j,\bX
_j^T)^T=\xi(O_j)$, $j=1,\dots,n$, where $O_1$,\dots, $O_n$ are
independent subjects which can belong to different components with probabilities
\[
p_j^m=\pr\bigl\{\kappa(O_j)=m \bigr\},\quad
m=1,\dots,M; \ j=1,\dots,n.
\]
(all mixing probabilities $p_j^m$ are known).

To describe completely the probabilistic behavior of the observed data
we need to introduce the distributions of $\varepsilon(O)$ and $\bX(O)$
for different components. Let us denote
\[
F_{\bX,m}(A)=\pr\bigl\{\bX(O)\in A\ |\ \kappa(O)=m \bigr\} \text{ for any
measurable } A\subseteq\R^d,
\]
and
\[
F_{\varepsilon,m}(A)=\pr\bigl\{\varepsilon(O)\in A |\ \kappa(O)=m \bigr\} \text{
for any measurable } A\subseteq\R.
\]
The corresponding probability densities $f_{\bX,m}$ and $f_{\varepsilon
,m}(x)$ are defined by
\[
F_{\bX,m}(A)=\int_A f_{\bX,m}(\bx)d\bx, \
F_{\varepsilon,m}(A)=\int_A f_{\varepsilon,m}(x)dx
\]
(for all measurable $A$).

The distribution of observed $\xi_j$ is a mixture of distributions of
components with the mixing probabilities $p_j^m$, e.g.
\[
\pr\{\bX_j\in A\}=\sum_{m=1}^M
p_j^m F_{\bX_j,m}(A)
\]
and the probability density $f_j(y,\bx)$ of $\xi_j=(Y_j,\bX_j^T)^T$ at
a point $(y,\bx^T)^T\in\R^{d+1}$ is\vadjust{\goodbreak}
\[
f_j(y,\bx)=\sum_{m=1}^M
p_j^m f_{\bX,m}(\bx)f_{\varepsilon,m}\bigl(y-
\bx^T\bb^m\bigr).
\]
In what follows we will discuss two approaches to the estimation of the
parameters of interest $\bb^k$, for a fixed $k\in\{1,\dots,M\}$.

The first one is the nonparametric approach. Under this approach we do
not need to know the densities $f_{\bX,m}$ and $f_{\varepsilon,m}$.
Moreover we even do not assume the existence of these densities. The
estimates are based on some modification of the least squares tehnique
proposed in \cite{Liubashenko}.

In the second, parametric approach we assume that the densities of
components are known up to some unknown
nuisance parameters $\vartheta_m\in\varTheta\subseteq\R^L$:
\begin{equation}
\label{EqParamM} f_{\bX,m}(\bx)=f(\bx;\vartheta_m),\qquad
f_{\varepsilon,m}(x)=f_{\varepsilon
}(x;\vartheta_m).
\end{equation}
In the most popular parametric \textit{normal mixture model} these
densities are normal, i.e.
\begin{equation}
\label{EqNormalMix} f_{\varepsilon,m}\sim N\bigl(0,\sigma^2_m
\bigr),\qquad f_{\bX,m}(\bx)\sim N(\mu _m,\varSigma_m),
\end{equation}
where
$\mu_m\in\R^d$ is the mean of $\bX$ for the $m$-th component and $\varSigma
_m\in\R^{d\times d}$ is its covariance matrix. All the parameters are
usually unknown.
So, in this case the unknown nusance parameters are
\[
\vartheta_m=\bigl(\mu_m,\varSigma_m,
\sigma_m^2\bigr),\quad m=1,\dots,M.
\]
\section{Generalized least squares estimator}\label{SectNonp}

Let us consider the nonparametric approach to the estimation of the
regression coefficients developed in
\cite{Liubashenko}. It is based on the minimization of weighted least squares
\[
J_{k;n}(\bb)\bydef{1\over n}\sum
_{j=1}^n a_{j;n}^k
\Biggl(Y_{j}-\sum_{i=1}^d
b_iX_{j}^i \Biggr)^2,
\]
over all possible $\bb=(b_1,\dots,b_d)^T\in\R^d$.

Here $\ba^k=(a_{1;n}^k,\dots,a_{n;n}^k)$ are the minimax weights for
estimation of the $k$-th component's distribution. They are defined by
\begin{equation}
\label{EqDefa} a_{j;n}^k={1\over
\det\bGamma_n}\sum
_{m=1}^M(-1)^{k+m}\gamma_{mk;n}p_{j}^m,
\end{equation}
where $\gamma_{mk;n}$ is the $(mk)$-th minor of the matrix
\[
\bGamma_n= \Biggl({1\over n}\sum
_{j=1}^n p_j^lp_j^i
\Biggr)_{l,i=1}^M,
\]
see \cite{Liubashenko,Maiboroda:StatisticsDNA} for details.

Define
$\bX\bydef(X_j^i)_{j=1,\dots,n;\ i=1,\dots,d}$ to be the $n\times
d$-matrix of observed regressors, $\bY\bydef(Y_1,\dots,Y_N)^T$ be
the vector of observed responses,
$\bA\bydef\diag(a_{1;n}^k,\dots,a_{n;n}^k)$\vadjust{\goodbreak} be the diagonal
weights matrix for estimation of $k$-th component. Then
the stationarity condition
\[
{\partial J_{k;n}(\bb)\over\partial\bb}=0
\]
has the unique solution in $\bb$,
\begin{equation}
\label{EqEsimDef} \hat\bb^{\mathit{LS}}(k,n)\bydef \bigl(\bX^T\bA\bX
\bigr)^{-1}\bX^T\bA\bY,
\end{equation}
if the matrix $\bX^T\bA\bX$ is nonsingular.

Note that the weight vector $\ba^k$ defined by (\ref{EqDefa}) contains
negative weights, so\break $\hat\bb_{\mathit{LS}}(k,n)$ is not allways the point of
minimum of $J_{k;n}(\bb)$. But in what follows we will consider $\hat\bb
_{\mathit{LS}}(k,n)$ as a generalized least squares estimate for $\bb^k$.

The asymptotic behavior of $\hat\bb_{\mathit{LS}}(k,n)$ as $n\to\infty$ was
investigated in \cite{Liubashenko}. To describe it we will need some
additional notation.

Let us denote by
\[
\bD^{(m)}\bydef\M \bigl[\bX(O)\bX^T(O)\ | \ \kappa(O)=m
\bigr]
\]
the matrix of second moments of regressors for the $m$-th component.

The consistency conditions for the estimator $\hat\bb^{\mathit{LS}}(k,n)$ are
given by the following theorem.

\begin{thm}[Theorem 1 in \cite{Liubashenko}]\label{ThConsist}
Assume that
\begin{enumerate}
\item[1.] $\bD^{(m)}$ and $\sigma_m^2$ are finite for
all $m=1,\dots,M$.
\item[2.] $\bD^{(k)}$ is nonsingular.
\item[3.] There exists $C>0$ such that $\det\bGamma_n>C$ for all $n$
large enough.
\end{enumerate}

Then $\hat\bb^{\mathit{LS}}(k,n)\inprob\bb^{(k)}$ as $n\to\infty$.
\end{thm}

Let us introduce the following notation.
\begin{align}
D^{\mathit{is}(m)}&\bydef\M\bigl[ X^i(O)X^s(O)|\kappa(O)=m\bigr],\nonumber\\
\bL^{\mathit{is}(m)}&\bydef\bigl(\M\bigl[ X^i(O)X^s(O)X^q(O)X^l(O)|
\kappa(O)=m\bigr]\bigr)_{l,q=1}^d,\nonumber\\
%
\bM^{\mathit{is}(m,p)}&\bydef\bigl(D^{\mathit{il}(m)}D^{sq(p)}
\bigr)_{l,q=1}^d,\nonumber\\
%
\label{EqAlpha} \alpha^{(k)}_{s,q}&=\lim_{n\to\infty}
{1\over n}\sum_{j=1}^n
\bigl(a_{j;n}^k\bigr)^2p_j^sp_j^q
\end{align}
(if this limit exists),
\[
\alpha^{(k)}_{s}=\lim_{n\to\infty}
{1\over n}\sum_{j=1}^n
\bigl(a_{j;n}^k\bigr)^2p_j^s=
\sum_{q=1}^M\alpha^{(k)}_{s,q}.
\]
The following theorem provides conditions for the asymptotic normality
and describes the dispersion matrix of the estimator $\hat\bb^{\mathit{LS}}(k,n)$.
\begin{thm}[Theorem 2 in \cite{Liubashenko}]\label{ThCLT}
Assume that
\begin{enumerate}
\item[1.] $\M[(X^i(O))^4\ |\ \kappa(O)=m]<\infty$ and
$\M[(\varepsilon(O))^4\ |\ \kappa(O)=m]<\infty$
for all $m=1,\dots,M$, $i=1,\dots,d$.
\item[2.] Matrix $\bD=\bD^{(k)}$ is nonsingular.
\item[3.] There exists $C>0$ such that $\det\bGamma_n>C$ for all $n$
large enough.
\item[4.] For all $s$,$q=1,\dots,M$ there exist $\alpha_{s,q}^{(k)}$ defined
by (\ref{EqAlpha}).
\end{enumerate}

Then $\sqrt{n}(\hat\bb^{\mathit{LS}}(k,n)-\bb^{(k)})\weak N(0,\bV)$,
where
\begin{align}
\label{EqDefV} \bV&\bydef\bD^{-1}\bSigma\bD^{-1},\\
\bSigma&=\bigl(\varSigma^{\mathit{il}}\bigr)_{\mathit{il}=1}^d,\nonumber\\
\label{varequal2} \varSigma^{\mathit{il}} &= \sum_{s=1}^M
\alpha^{k}_{s} \bigl( D^{\mathit{il}(s)}\sigma_s^2
+ \bigl(\bb^{s} - \bb^{k}\bigr)^T
\bL^{\mathit{il}(s)} \bigl(\bb^{s} - \bb^{k}\bigr) \bigr)
\nonumber
\\
&\quad- \sum_{s=1}^M \sum
_{m=1}^M \alpha^{k}_{s,m} \bigl(
\bb^{s} - \bb^{k}\bigr)^T \bM^{\mathit{il}(s,m)}
\bigl(\bb^{m} - \bb^{k}\bigr).
\end{align}
\end{thm}

\section{Parametric approach}\label{SectParam}

In this section we discuss the parametric approach to the estimation
of $\bb^k$ based on papers \cite{Faria,GrunLeisch}. We will assume that the
representation (\ref{EqParamM}) holds with some unknown $\vartheta_m\in
\varTheta\subseteq\R^L$, $m=1,\dots,M$.

Then the set of all unknown parameters $\tau=(\bb^k,\beta,\vartheta)$
consists of
\[
\beta=\bigl(\bb^m,m=1,\dots,M,m\neq k\bigr)
\]
and
\[
\vartheta=(\vartheta_1,\dots,\vartheta_M).
\]
Here $\bb^k$ is our parameter of interest, $\beta$ and $\vartheta$ are
the nuisance parameters.

In this model the log-likelihood for the unknown $\tau$ by the sample
$\varXi_n$ can be defined as
\[
L(\tau)=\sum_{j=1}^n L(
\xi_j,\bp_j,\tau),
\]
where
$\bp_j=(p_j^1,\dots,p_j^M)^T$,
\[
L(\xi_j,\bp_j,\tau)=\ln \Biggl(\sum
_{m=1}^M p_j^m
f_{X,m}(\bX_j;\vartheta_m)f_{\varepsilon,m}
\bigl(\bY_j-\bX_j^T\bb^m\bigr)
\Biggr).
\]
The general maximum likelihood estimator (MLE)
$\hat\tau^{\mathit{MLE}}_n\,{=}\,(\hat\bb^{k,{\mathit{MLE}}},\hat\beta^{\mathit{MLE}},\hat\vartheta^{\mathit{MLE}})$
for $\tau$ is defined as
\[
\hat\tau^{\mathit{MLE}}_n=\argmax_\tau L(\tau),
\]
where the maximum is taken over all possible values of $\tau$. Unfortunately,
this estimator is not applicable to most common parametric mixture
models, since
the log-likelihood $L(\tau)$ usually is not bounded on the set of all
possible $\tau$.

For example, it is so in the normal mixture model (\ref{EqNormalMix}).
Really, in this model
$L(\tau)\to\infty$ as $\sigma_1^2\to0$ and $Y_1-\bX_1^T\bb^1=0$ with
all other parameters being
arbitrary fixed.

The usual way to cope with this problem is to use the one-step MLE, which
can be considered as one iteration of the Newton--Raphson algorithm of
approximate
calculation of MLE, starting from some pilot estimate (see \cite{Shao},
section 4.5.3). Namely, let
$
\hat\tau_n^{(0)}
$
be some pilot estimate for $\tau$.
Let us consider $\tau$ as a vector of dimension
$P=d\times M+M\times L$ and denote its entries by $\tau_i$:
\[
\tau=(\tau_1,\dots,\tau_P).
\]
Denote the gradient of $L(\tau)$ by
\[
s_n(\tau)={\partial L(\tau)\over\partial\tau}= \biggl( {\partial L(\tau)\over\partial\tau_1},
\dots,{\partial L(\tau)\over
\partial\tau_P} \biggr)^T
\]
and the Hessian of $L(\tau)$ by
\[
\gamma_n(\tau)= \biggl( {\partial L(\tau)\over\partial\tau_i\tau_l}
\biggr)_{i,l=1}^P.
\]
Then the one-step estimator for $\tau$ starting from $\hat{\tau^{(0)}}$
is defined as
\[
\hat\tau_n^{\mathit{OS}}=\hat\tau^{(0)}- \bigl(
\gamma_n\bigl(\hat\tau^{(0)}\bigr) \bigr)^{-1}s_n
\bigl(\hat\tau^{(0)}\bigr).
\]
Theorem 4.19 in \cite{Shao} provides general conditions under which
$\hat\tau_n^{\mathit{OS}}$ constructed by an i.i.d. sample is
consistent, asymptotically normal and asymptotically efficient.\footnote
{Note that in our setting the sample $\varXi_n$ is not
an i.i.d. sample. But one can consider it as i.i.d. if the vectors of
concentrations $(p_j^1,\dots,p_j^M)$
are generated by some stochastic mechanism as i.i.d. vectors. See \cite
{MS2015-2} for an example of such \textit{stochastic concentrations}
models.} The limit distribution of the normalized one-step estimate is
the same as of the consistent version of MLE.

So, if the assumptions of theorem 4.19 (or other analogous statement)
hold, there is no need to use an iterative procedure to derive an
estimate with asymptotically optimal performance. But on samples of
moderate size $\hat\tau^{\mathit{OS}}_n$ can be not good enough.

Another popular way to obtain a stable estimate for $\tau$ is to use
some version of EM-algorithm. A general EM-algorithm is an iterative
procedure for approximate calculation of maximum likelihood estimates
when information on some variables is missed. We describe here only the
algorithm which calculates EM estimates $\hat\tau^{\mathit{EM}}_n$ under the
normal mixture model assumptions (\ref{EqNormalMix}), cf. \cite
{Faria,GrunLeisch}.

The algorithm starts from some pilot estimate
\[
\hat\tau^{(0)}=\bigl(\hat\bb^{(0)}_m,\hat
\sigma^{2(0)}_m,\hat\mu _m^{(0)},\hat
\varSigma_m^{(0)},m=1,\dots,M\bigr)
\]
for the full set of the model parameters.\vadjust{\goodbreak}

Then for $i=1,2,\dots$ the estimates are iteratively recalculated in
the following way.

Assume that on the $i$-iteration estimates
$\hat\bb^{(i)}$, $\hat\sigma^{2(i)}_m$, $\hat\mu_m^{(i)}$, $\hat
\varSigma_m^{(i)}$, $m=1,\dots,M$ are obtained. Then
the $i$-th stage weights are defined as
\begin{equation}
\label{EqWdef} w_j^{m(i)}=w_j^m
\bigl(\xi_j,\hat\tau^{(i)}\bigr) = {
p_j^mf_{\bX,m}(\bX_j;\hat\mu_m^{(i)},\hat\varSigma_m^{(i)})
f_{\varepsilon,m}(Y_j-\bX_j^T\hat\bb^{m(i)};\hat\sigma_m^{2(i)})
\over
\sum_{l=1}^M
p_j^lf_{\bX,l}(\bX_j;\hat\mu_l^{(i)},\hat\varSigma_l^{(i)})
f_{\varepsilon,l}(Y_j-\bX_j^T\hat\bb^{l(i)};\hat\sigma_l^{2(i)})
}
\end{equation}
(note that $w_j^{m(i)}$ is the posterior probability
$\pr\{\kappa_j=m\ |\xi_j \}$ calculated for $\tau=\hat\tau^{(i)}$).

Let $\bar w^m=\sum_{j=1}^n w_j^{m(i)}$.
Then the estimators of the $i+1$ iteration are defined as
\begin{align*}
\hat\mu_m^{(i+1)}&={1\over\bar w^m}\sum
_{j=1}^n w_j^{m(i)}
\bX_j,\\
\hat\varSigma_m^{(i+1)}&={1\over\bar w^m}\sum
_{j=1}^n w_j^{m(i)} \bigl(
\bX_j-\hat\mu_m^{(i)}\bigr) \bigl(
\bX_j-\hat\mu_m^{(i)}\bigr)^T,\\
\hat\bb^{m(i+1)}&= \Biggl( \sum_{j=1}^n
w_j^{m(i)}\bX_j\bX_j^T
\Biggr)^{-1} \sum_{j=1}^n
w_j^{m(i)}Y_j\bX_j,\\
\hat\sigma_m^{2(i+1)}&={1\over\bar w^m}\sum
_{j=1}^n w_j^{m(i)}
\bigl(Y_j-\bX_j^T\hat b^{m(i)}
\bigr)^2.
\end{align*}
The iterations are stopped when some stopping condition is fulfilled.
For example, it can be
\[
\bigl\|\hat\tau^{(i+1)}-\hat\tau^{(i)}\bigr\|<\delta,
\]
where $\delta$ is a prescribed target accuracy.

It is known that this procedure provide stable estimates which (for
sample large enough) converge
to the point of local minimum of $L(\tau)$ which is the closest to the
pilot estimator $\hat\tau^{(0)}$.

So, this estimator can be considered as an approximate version of a
root of likelihood equation estimator (RLE).

The asymptotic behavior of $\hat\tau_n^{\mathit{OS}}$ and $\hat\tau_n^{\mathit{EM}}$
can be described in terms of Fisher's information matrix
$\bI^*(n,\tau)=(I^*_{\mathit{il}}(n,\tau))_{i,l=1}^P$, where
\[
I^*_{\mathit{il}}(n,\tau)=\sum_{j=1}^nI_{\mathit{il}}(
\bp_j,\tau), \ I_{\mathit{il}}(\bp,\tau)=\M{\partial L(\xi_\bp,\bp,\tau)\over\partial\tau
_i}
{\partial L(\xi_\bp,\bp,\tau)\over\partial\tau_l},
\]
where $\bp=(p^1,\dots,p^m)$, $\xi_\bp$ is a random vector with the pdf
\begin{equation}
\label{Eqfbp} f_\bp(y,\bx;\tau)=\sum_{m=1}^M
p^m f(\bx;\vartheta_m)f_{\varepsilon
}\bigl(y-
\bx^T\bb^m;\vartheta_m\bigr).
\end{equation}

Under the regularity \querymark{Q1}assumptions (RR) of theorem 70.5 in
\cite{BorovkovMs},
\begin{equation}
\label{EqMLEAsNorm} \bigl(\bI^*(n,\tau) \bigr)^{1/2}\bigl(\hat
\tau^{\mathit{MLE}}_n-\tau\bigr)\weak N(0,\EE),
\end{equation}
where $\EE$ is the $R\times R$ unit matrix.

Assumptions (RR) include the assumption of likelihood boundedness, so
they do not hold for the normal mixture model. But if the
pilot estimate $\hat\tau_n^{(0)}$ is $\sqrt{n}$-consistent, one needs
only a local version of (RR) to derive asymptotic
normality of $\hat\tau_n^{\mathit{OS}}$ and $\hat\tau_n^{\mathit{EM}}$, i.e. (RR) must
hold in some neighborhood of the true value of
estimated $\tau$. These local (RR) hold for the normal mixture model if
$\sigma_m^2>0$ and $\varSigma_m$ are nonsingular for all $m=1,\dots,M$.

To use all these results for construction of an estimator for $\bb^k$
we need $\sqrt{n}$-consistent pilot estimators for the parameter of
interest and nuisance parameters. They can be derived by the
nonparametric technique considered in Section \ref{SectNonp}. To
construct confidence ellipsoids we will also need estimators for the
dispersion matrix $\bV$ from (\ref{EqDefV}) in the nonparametric case
and estimators for the information matrix $\bI^*(n,\tau)$ in the
parametric case.
These estimators are discussed in the next section.

\section{Estimators for nuisance parameters and normalizing
matrices}\label{SecEstN}

Let us start with the estimation of the dispersion matrix $\bV$ in
Theorem \ref{ThCLT}.
In fact, we need to estimate consistently the matrices $\bD$ and
$\bSigma$.

Note that
$\bD=\bD^{(k)}=(D^{\mathit{is}(k)})_{i,s=1}^d$, where
\[
D^{\mathit{is}(k)}=\M\bigl[X^i(O)X^s(O)\ | \ \kappa(O)=k
\bigr].
\]
By theorem 4.2 in \cite{Maiboroda:StatisticsDNA},
\begin{equation}
\label{EqEstD} \hat D_n^{\mathit{is}(k)}={1\over n}\sum
_{j=1}^n a_j^k
X_j^iX_j^s
\end{equation}
is a consistent estimate for $D^{\mathit{is}(k)}$ if
$\M[\|\bX(O)\|^2\ |\ \kappa(O)=m]<\infty$
for all $m=1,\dots,M$ and assumption 3 of Theorem \ref{ThConsist} holds.

So one can use $\hat\bD_n^{(k)}=(\hat D_n^{\mathit{is}(k)})_{i,s=1}^d$ as a consistent
estimate for $\bD$ if the assumptions of Theorem \ref{ThConsist} hold.

Similarly, $\bL^{\mathit{is}(m)}$ can be estimated consistently by
\begin{equation}
\label{EqEstL} \hat\bL^{\mathit{is}(m)}_n={1\over n}\sum
_{j=1}^n a_j^m
X_j^iX_j^s\bX_j
\bX_j^T
\end{equation}
under the assumptions of Theorem \ref{ThCLT}.

The same idea can be used to estimate $\sigma^2_m$ by
\begin{equation}
\label{EqEstSigma} \hat\sigma^{2(0)}_{m;n}={1\over n}
\sum_{j=1}^n a_j^m
\bigl(Y_j-\bX_j^T\hat \bb^{\mathit{LS}}(s,n)
\bigr)^2.
\end{equation}
The coefficients $\alpha_{s,q}^{(k)}$ can be approximated by
\[
\hat\alpha^{(k)}_{s,q}={1\over n}\sum
_{j=1}^n \bigl(a_{j;n}^k
\bigr)^2p_j^sp_j^q.
\]
Now replacing true $\bD^{(m)}$, $\bL^{\mathit{is}(m)}$, $\bb^{m}$, $\sigma^2_m$
and $\alpha^{(k)}_{s,q}$ in
formula (\ref{varequal2}) by their estimators
$\hat\bD^{(m)}_n$, $\hat\bL^{\mathit{is}(m)}_n$, $\hat\bb^{\mathit{LS}}(m,n)$, $\hat
\sigma^2_{m,n}$ and $\hat\alpha^{(k)}_{s,q}$,
one obtains a consistent estimator $\hat\varSigma_n$ for $\varSigma$.

Then
\begin{equation}
\label{EqEstV} \hat\bV_n=\hat\bD_n^{-1}\hat
\varSigma_n\hat\bD_n^{-1}
\end{equation}
is a consistent estimator for $\bV$.

To get the pilot estimators for the normal mixture model one can use
the same approach. Namely, we define
\[
\hat\mu_{m,n}^{(0)}={1\over n}\sum
_{j=1}^n a_j^m
\bX_j,\qquad \hat\varSigma_{m,n}^{(0)}=
{1\over n}\sum_{j=1}^n
a_j^m(\bX_j-\hat\mu _{m,n}) (
\bX_j-\hat\mu_{m,n})^T
\]
as estimates for $\mu_m$ and $\varSigma_m$.

By theorem 4.3 from \cite{Maiboroda:StatisticsDNA}, $\hat\mu_{m,n}^{(0)}$,
$\hat\varSigma_{m,n}^{(0)}$ and $\hat\sigma^{2(0)}_{m,n}$ are $\sqrt
{n}$-consistent estimators for the corresponding parameters
of the normal mixture model. This allows one to use them as pilot
estimators for the one-step and EM estimators.

Now let us consider estimation of the Fisher information matrix in the
case of normal mixture model. Define
\begin{align}
\label{EqII} \hat I_{\mathit{il}}(n,\tau)&=\sum_{j=1}^n
{\partial L(\xi_j,\bp_j,\tau)\over\partial\tau_i}{\partial L(\xi
_j,\bp_j,\tau)\over\partial\tau_l},\\
\label{EqEmpInf} \hat\bI(n,\tau)&=\bigl(\hat I_{\mathit{il}}(n,\tau)
\bigr)_{i,l=1}^R,\qquad \hat\bI(n)=\hat\bI(n,\hat\tau)
\end{align}
where $\hat\tau$ can be any consistent estimator for $\tau$ (e.g. $\hat
\tau^{\mathit{OS}}_n$ or $\hat\tau^{\mathit{EM}}_n$).

In the normal mixture model we will denote
$\tau_{(l)}=(\bb^{(l)},\mu_l,\varSigma_l,\sigma_l^2)$, i.e. the set of
all unknown parameters which describe the $l$-th mixture component.

\begin{thm}\label{ThInfConsist}
Assume that the normal mixture model \querymark{Q2}is taken 
and
\begin{enumerate}
\item[1.] $\sigma^2_m>0$, $\varSigma_m$ are nonsingular for all $m=1,\dots,M$;
\item[2.] There exist $c>0$ such that for all $j=1,\dots,n$, $m=1,\dots,M$,
$n=1,2,\dots$
\[
p_j^m>c.
\]
\item[3.] $\tau^{(l)}\neq\tau^{(m)}$ for all $l\neq m$, $l,m=1,\dots,M$.
\end{enumerate}

Then
\begin{enumerate}
\item[1.] There exist $0<c_0<C_1<\infty$ such that
\[
c_on\le\bigl\|\bI^*(n,\tau)\bigr\|\le C_1 n
\]
for all $n=1,2,\dots$.
\item[2.] ${1\over n}\|\bI^*(n,\tau)-\hat\bI(n)\|\to0$
in probability as $n\to\infty$.\vadjust{\goodbreak}
\end{enumerate}
\end{thm}

\begin{note} Here and below for any square matix $\bI$ the symbol $\|
\bI\|$ means the operator norm of $\bI$, i.e.
\[
\|\bI\|=\sup_{\bu:\ \|\bu\|=1}\bigl|\bu^T\bI\bu\bigr|.
\]
\end{note}

\begin{proof} 1. At first we will show that
\begin{equation}
\label{EqInon0} \bu^T\bI(\bp,\tau)\bu>0
\end{equation}
for any $\tau$ and any $\bu\in\R^P$ with $\|\bu\|=1$
and for any $\bp=(p^1,\dots,p^M)$ with $p^m>c$ for all $m=1,\dots,M$.

Recall that $\tau=(\tau_{(1)},\dots,\tau_{(M)})$, where $\tau_{(m)}$
corresponds to the parameters describing the $m$-th component. Let us
divide $\bu$ into analogous blocks $\bu=(\bu_{(1)}^T,\dots,\break\bu_{(M)}^T)^T$.

Then
\begin{align*}
\bu^T\bI(\bp,\tau)\bu &=\M{\bu^T}{\partial\over\partial\tau}L(
\xi_\bp,\bp,\tau) \biggl({\partial\over\partial\tau}L(
\xi_\bp,\bp,\tau) \biggr)^T \bu\\
&=\M \biggl( {\bu^T}{\partial\over\partial\tau}L(\xi_\bp,
\bp,\tau) \biggr)^2\\
&=\M \Biggl(\sum_{m=1}^M {
\bu_{(m)}^T}{\partial\over\partial\tau_{(m)}}L(\xi_\bp,
\bp,\tau) \Biggr)^2.
\end{align*}
Note that
$
{\partial\over\partial\tau_{(m)}}L(\xi_\bp,\bp,\tau)
$
can be represented as
\begin{equation}
\label{EqLA} {\partial\over\partial\tau_{(m)}}L(\xi_\bp,\bp,\tau) =\bA(
\tau_{(m)},\xi_\bp) {p^m\varphi_{\tau_{(m)}}(\xi_\bp)\over f_\bp(\xi_\bp;\tau)}
\end{equation}
where $\varphi_{\tau_{(m)}}$ is the normal pdf of the observation $\xi$
from the $m$-th component, $f_\bp$ is the pdf of the mixture defined by
(\ref{Eqfbp}),
$\bA(\tau_{(m)},\xi_\bp)$ is a vector with entries which are polynomial
functions from the entries of $\xi_\bp$.

Then
\begin{equation}
\label{Eq5*} \bu^T\bI(\bp,\tau)\bu= \M \Biggl( \sum
_{m=1}^M p^m \bu_{(m)}^T
\bA(\tau_{(m)},\xi_\bp)\varphi_{\tau_{(m)}}(
\xi_\bp) \Biggr)^2 {1\over f_\bp(\xi_\bp;\tau)^2}.
\end{equation}
Note that by the assumptions 1 and 2 of the theorem $f_\bp(\xi;\tau)>0$
for all $\xi\in\R^{d+1}$.
Then $\bu_{(m)}^T\bA(\tau_{(m)},\xi_\bp)$ are polynomials of $\xi_\bp$
and $\varphi_{\tau_{(m)}}(\xi_\bp)$ are \querymark{Q3}exponentials of
different (due to assumption 3) and nonsingular (due to assumption 1)
quadratic forms of $\xi_\bp$.

Suppose that\querymark{Q4} $\bu^T\bI(\bp,\tau)\bu$ for some $\bu$ with
$\|\bu\|=1$. Then (\ref{Eq5*}) implies
\begin{equation}
\label{Eq2*} \bu_{(m)}^T\bA(\tau_{(m)},
\xi_\bp)=0 \text{ a.s.}
\end{equation}
for all $m=1,\dots,M$.

On the other hand, (\ref{Eq2*}) implies
\[
\M \bigl(\bu_{(m)}^T\bA(\tau_{(m)},
\xi_\bp)\varphi_{\tau_{(m)}}(\xi_\bp )
\bigr)^2=\bu_{(m)}^T\bI_{\tau_{(m)}}
\bu_{(m)}=0,
\]
where $\bI_{\tau_{(m)}}$ is the Fisher information matrix for the
unknown $\tau_{(m)}$ by one observation from the $m$-th component. By
the assumption 1, $\bI_{\tau_{(m)}}$ is nonsingular, so $\bu_{(m)}=0$
for all $m=1,\dots,M$. This contradicts the assumption $\|\bu\|=1$.

So, by contradiction, (\ref{EqInon0}) holds.
Since $\bu^T\bI(\bp,\tau)\bu$ is a continuous function on the compact
set of $\bu:\ \|\bu\|=1$ and $\bp$ satisfying assumption 2, from (\ref
{EqInon0}) we obtain
$\bu^T\bI(\bp,\tau)\bu>c_0$ for some $c_0>0$.
On the other hand, the representation (\ref{EqLA}) implies $\|\bI(\bp
,\tau)\|<C_1$

Then from
$\bI^*(n,\tau)=\sum_{j=1}^n\bI(\bp_j,\tau)$
we obtain the first statement of the theorem.

2. To prove the second statement note that
by the law of large numbers
\begin{align*}
\Delta_n(\tau)&={1\over n}\bigl(\hat\bI(n,\tau)-
\bI^*(n,\tau)\bigr)\\
&={1\over n} \sum_{j=1}^n
\biggl[ {\partial L(\xi_j,\bp_j,\tau)\over\partial\tau} \biggl({\partial L(\xi_j,\bp_j,\tau)\over\partial\tau}
\biggr)^T \\
&\quad- \M {\partial L(\xi_j,\bp_j,\tau)\over\partial\tau} \biggl({\partial L(\xi_j,\bp_j,\tau)\over\partial\tau}
\biggr)^T \biggr]\\
&\inprob0,\quad \text{ as } n\to\infty,
\end{align*}
since
\[
\M \biggl\llVert {\partial L(\xi_j,\bp_j,\tau)\over\partial\tau} \biggr\rrVert ^4\le C<
\infty
\]
for all $j$.

Let $B\subseteq\R^P$ be any open bounded neighborhood of $\tau$. Note that
\[
\M\sup_{\tau\in B} \biggl\llVert {\partial\over\partial\tau}
{\partial L(\xi_j,\bp_j,\tau)\over\partial\tau} \biggl({\partial L(\xi_j,\bp_j,\tau)\over\partial\tau} \biggr)^T
\biggr\rrVert <C_2<\infty.
\]
From this together with $\Delta_n(\tau)\inprob0$ we obtain
\[
\sup_{\tau\in B}\bigl\|\Delta_n(\tau)\bigr\|\inprob0
\]
(applying the same technique as in lemma 5.3 from \cite{Shao}).

The last equation together with $\hat\tau_n\inprob\tau$ implies the
second statement of the Theorem.
\end{proof}

\section{Confidence ellipsoids for $\bb^k$}\label{SecCE}

Let $\varXi_n$ be any random dataset of size $n$ with distribution
dependent of an unknown parameter $\bb\in\R^d$. Recall
that a set $B_\alpha=B_\alpha(\varXi_n)\subset\R^d$ is called an
asymptotic confidence set of the significance level
$\alpha$ if
\[
\lim_{n\to\infty}\pr\bigl\{\bb\notin B_\alpha(
\varXi_n) \bigr\}=\alpha.
\]

We will construct confidence sets for the vector of regression
coefficients $\bb=\bb^k$ by the sample from a mixture $\varXi_n$ described
in Section \ref{SecModel}. In the nonparametric case the set will be
defined by statistics of the form
\[
S^{\mathit{LS}}(\beta)=n\bigl(\beta-\hat\bb^{\mathit{LS}}(k,n)
\bigr)^T\hat\bV_n^{-1}\bigl(\beta-\hat\bb
^{\mathit{LS}}(k,n)\bigr).
\]
In the parametric case we take the matrix $\hat\bI(n)$ defined by (\ref
{EqEmpInf}) and consider its inverse matrix $\hat\bI
(n)^{-1}=(I^{-}(i,m))_{i,m=1}^R$.

Note that by (\ref{EqII}) and (\ref{EqEmpInf}) the elements $\hat
I_{im}$ of $\hat\bI(n)$ correspond to coordinates $\tau_i$ and $\tau
_m$ of the vector of unknown parameters $\tau$. Let us take the set of
indices $l_m$, $m=1,\dots,d$ such that
$\tau_{l_m}=b_m^k$ and consider the matrix
\[
\bigl[\hat\bI(n)^{-1}\bigr]_{(k)}=\bigl(I^{-}(l_i,l_m)
\bigr)_{i,m=1}^d.
\]
So, the matrix $[\hat\bI(n)^{-1}]_{(k)}$ contains the elements of
$\hat\bI(n)^{-1}$ corresponding to $\bb^{(k)}$ only.

Then we invert this matrix once more:
\[
\hat\bI_k(n)^{+}= \bigl(\bigl[\hat\bI(n)^{-1}
\bigr]_{(k)} \bigr)^{-1}.
\]
This matrix is used to construct the statistics which defines the
confidence set:
\[
S^{\mathit{OS}}(\beta)=\bigl(\beta-\hat\bb^{\mathit{OS}}(k,n)
\bigr)^T \hat\bI_k(n)^{+}\bigl(\beta-\hat\bb
^{\mathit{OS}}(k,n)\bigr)
\]
or
\[
S^{\mathit{EM}}(\beta)=\bigl(\beta-\hat\bb^{\mathit{EM}}(k,n)
\bigr)^T \hat\bI_k(n)^{+}\bigl(\beta-\hat\bb
^{\mathit{EM}}(k,n)\bigr).
\]
Here $\hat\bb^{\mathit{OS}}(k,n)$ and $\hat\bb^{\mathit{EM}}(k,n)$ are the parts of the
estimators $\hat\tau^{\mathit{OS}}_n$ and
$\hat\tau^{\mathit{EM}}_n$ which esitmate $\bb^k$.

In what follows the symbol $\star$ means any of symbols $\mathit{LS}$, $\mathit{OS}$ or
$\mathit{EM}$. The confidence set $B^\star_\alpha(\varXi_n)$ is
defined by
\begin{equation}
\label{EqDefEll} B^\star_\alpha(\varXi_n)=\bigl\{
\beta\in\R^d:\ S^\star(\beta)\le Q^{\chi
_d^2}(1-\alpha)
\bigr\},
\end{equation}
where $Q^{\chi_d^2}(1-\alpha)$ is the $(1-\alpha)$-quantile of $\chi^2$
distribution with $d$ degrees
of freedom.

In the parametric case $\hat\bI_k(n)^{+}$ is a positively defined
matrice, so $B^\star_\alpha(\varXi_n)$
defined by (\ref{EqDefEll}) is the interior of an ellipsoid centered
at $\hat\bb^\star(k,n)$.

In the nonparametric case the matrix $\hat\bV_n$ can be not
positively defined for small $n$, so the set $B^{\mathit{LS}}_\alpha(\varXi_n)$ can
be unbounded. We will discuss some remedial actions for this problem in
Section \ref{SecSimul}.

\begin{thm}\label{ThConfNP}
Under the assumptions of Theorem \ref{ThCLT},
\[
\lim_{n\to\infty}\pr\bigl\{\bb^k\notin
B^{\mathit{LS}}_\alpha(\varXi_n) \bigr\}=\alpha.
\]
\end{thm}
\begin{proof} Theorem \ref{ThCLT} and consistency of $\hat\bV_n$
imply that $S^{\mathit{LS}}(\bb^k)\weak\chi_d^2$, so
\[
\pr\bigl\{\bb^k\notin B^{\mathit{LS}}_\alpha(
\varXi_n) \bigr\}=\pr\bigl\{S^{\mathit{LS}}\bigl(\bb^k
\bigr)>Q^{\chi
_d^2}(1-\alpha) \bigr\}\to\alpha
\]
as $n\to\infty$.
\end{proof}

\begin{thm}\label{ThConfParam}
Under the assumptions of Theorem \ref{ThInfConsist},
\[
\lim_{n\to\infty}\pr\bigl\{\bb^k\notin
B^{\mathit{OS}}_\alpha(\varXi_n) \bigr\}=\alpha.
\]
\end{thm}
\begin{proof}[Sketch proof] By theorem 70.5 from \cite{BorovkovMs} one
obtains the asymptotic normality of the local MLE estimate
\[
\hat\tau^{l\mathit{MLE}}_n=\argmax_{\tau\in D}L(\tau),
\]
where $D$ is a sufficiently small neighborhood of the true $\tau$. Then
the convergence
\begin{equation}
\label{EqSaNormOS} \bI^*(n,\tau)^{-1/2}\bigl(\hat\tau^{\mathit{OS}}_n-
\tau\bigr)\weak N(0,\EE)
\end{equation}
can be obtained from the asymptotic normality of $\hat\tau^{l\mathit{MLE}}_n$ by
the technique of theorem 14.19 from \cite{Shao}.

Let us denote
\[
S^{\mathit{OS}}_0(\beta)=\bigl(\beta-\hat\bb^{\mathit{OS}}(k,n)
\bigr)^T \bI_k(n)^{+}\bigl(\beta-\hat\bb
^{\mathit{OS}}(k,n)\bigr),
\]
where
$\bI_k(n)^{+}$ is the theoretical counterpart of $\hat\bI_k(n)^{+}$:
\[
\bI_k(n)^{+}= \bigl(\bigl[ \bI^*(n,\tau)^{-1}
\bigr]_{(k)} \bigr)^{-1}.
\]
Then by (\ref{EqSaNormOS}), $S^{\mathit{OS}}_0(\bb^{k})\weak\chi^2_d$.

Note that (\ref{EqSaNormOS}) and the first statement of Theorem \ref
{ThInfConsist} imply
\[
\zeta_n=\hat\bb^{\mathit{OS}}(k,n)-\bb^{k}=O_p
\bigl(n^{-1/2}\bigr).
\]
The second statement of Theorem \ref{ThInfConsist} implies
\[
{1\over n}\bigl\|\hat\bI_k(n)^{+}-
\bI_k(n)^{+}\bigr\|\inprob 0.
\]
So
\[
S^{\mathit{OS}}_0\bigl(\bb^{k}\bigr)-S^{\mathit{OS}}
\bigl(\bb^{k}\bigr) =\zeta_n^T
{1\over n}\bigl(\hat\bI_k(n)^{+}-
\bI_k(n)^{+}\bigr)\zeta_n\inprob 0
\]
and $S^{\mathit{OS}}(\bb^{k})\weak\chi^2_d$.

This completes the proof.
\end{proof}

\section{Results of simulations}\label{SecSimul}

We carried out a small simulation study to assess performance of the
parametric and nonperametric confidence intervals described above. A
two component mixture $(M=2)$
of simple regressions was simulated. The regression models were of the form
\begin{equation}
\label{EqSimpleRegr} Y=b_0^{\kappa}+b_1^{\kappa}X+
\varepsilon^{\kappa},
\end{equation}
where $X^k\sim N(\mu_k,\varSigma^2_k)$ and $Y$ are the observed regressor
and response, $\kappa$ is the unobserved number of components,
$\varepsilon^k$ is the regression error. The error $\varepsilon^k$ has
zero mean and variance $\sigma_k^2$.

The mixing probabilities were simulated by the following stochastic model:
\[
p_{j;N}^m = {u_{j}^m \over\varSigma_{s=1}^M{u_{j}^s}},
\]
where $u_{j}^m$ are independent uniformly distributed on $[0,1]$.

For each sample size $n$ we generated $1000$ samples. Parametric (EM)
and nonparametric (LS) confidence ellipsoids were constructed by each
sample. The parametric ellipsoids were based on EM-estimates which used
the LS-estimates as the pilot ones and $\hat\bI_k(n)^{+}$ as the matrix
for the quadratic form in $S^{\mathit{EM}}$.

The nonparametric confidence ellipsoids were based on the LS-estimates.
As it was mentioned in Section \ref{SecCE}, the matrix $\hat\bV_n$ can
be not positively defined. Then the corresponding confidence set will
be unbounded. In the case of simple regression (\ref{EqSimpleRegr})
this drawback can be cured by the use of improved weights $b_j^{+}$
defined in \cite{Kub1} instead of $a_j^k$ in (\ref{EqEstD})--(\ref
{EqEstSigma}). This technique was used in our simulation study.

All the ellipsoids were constructed with the nominal confidence level
$\alpha=0.05$. The frequencies of covering true $\bb^{k}$ by the
constructed ellipsoids and their mean volume were calculated in each
simulation experiment.

\begin{experiment}
The values of parameters for this experiment are presented in Table~\ref{TabPar1}.
The errors $\varepsilon^k$ were Gaussian. This is a
``totally separated'' model in which the observations can be visually
divided into two groups corresponding to different mixture components
(see the left panel at Fig. \ref{FigScat12}).
\begin{figure}[t]
\includegraphics{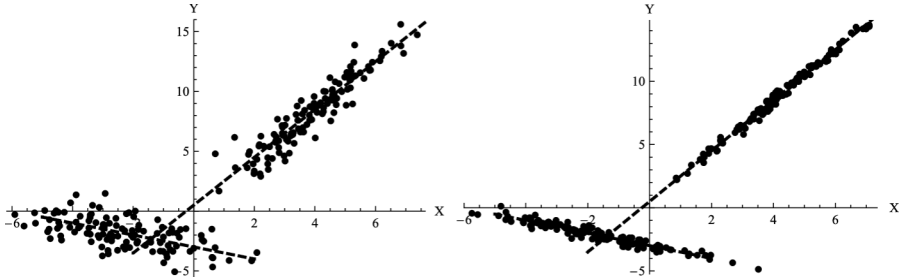}
\caption{Typical scatterplots of data in Experiment 1 (left) and
Experiment 2 (right)}\label{FigScat12}
\end{figure}

\begin{table}[t!]
\caption{Parameters for simulation in Experiments 1 and 2}\label{TabPar1}
%
\begin{tabular}{ | l | c | r |}
\hline
$k$ & 1 & 2 \\ \hline
$\mu_k$ & $-$2 & 4 \\ \hline
$\varSigma_k$ & 3 & 2 \\ \hline
$\sigma_k$ & 1 & 1 \\ \hline
$b_0^k$ & $-$3 & 0.5 \\ \hline
$b_1^k$ & $-$0.5 & 2 \\
\hline
\end{tabular}
%
\end{table}

\begin{table}[t!]
\caption{Experiment 1 results ($k$ is the number of component)}\label{TabEx1}
\begin{tabular}{|l|c|c|c|c|}
\multicolumn{5}{c}{Covering frequencies}\\
\hline
\tmultirow{2}{*}{$n$} &\multicolumn{2}{c|}{LS} & \multicolumn{2}{c|}{EM}\\
\cline{2-5}
& $k=1$ & $k=2$ & $k=1$ & $k=2$\\
\hline
100 &0.954 &0.821 & 0.946 & 0.951\\
$10^3$& 0.975 & 0.914 & 0.947 & 0.952\\
$10^4$&0.988 & 0.951 & 0.95 & 0.95\\
$10^5$& 0.951 & 0.963 & 0.952 & 0.953\\
$10^6$ & 0.936 & 0.949 & 0.936 & 0.951\\
\hline
\end{tabular}
\begin{tabular}{|l|c|c|c|c|}
\multicolumn{5}{c}{Average volume of ellipsoids}\\
\hline
\tmultirow{2}{*}{$n$} &\multicolumn{2}{c|}{LS} & \multicolumn{2}{c|}{EM}\\
\cline{2-5}
& $k=1$ & $k=2$ & $k=1$ & $k=2$\\
\hline
100 &$298* 10^6$ &$262* 10^6$ & 2.177543 & 0.243553\\
$10^3$& 1364 & 394 & 0.186303 & 0.021234\\
$10^4$&0.476327 & 0.317320 & 0.018314 & 0.002062\\
$10^5$&0.041646 & 0.030047 & 0.001845 & 0.000207\\
$10^6$ & 0.004121 & 0.002988 & 0.000185 & 0.000021\\
\hline
\end{tabular}
\vspace*{6pt}
\end{table}

Covering frequencies and mean volumes of the ellipsoids for different
sample sizes $n$ are presented in Table~\ref{TabEx1}.
They demonstrate sufficient accordance with the nominal significance
level for sample sizes greater then 1000. Extremely large mean volumes
for the LS-ellipsoids are due to poor performance of the estimates $\hat
\bV_n$ for small and moderate sample sizes $n$.

The parametric confidence sets are significantly smaller then the
nonparametric ones.
\end{experiment}

\begin{experiment}
To see how the standard deviations of regression errors affect the
performance of our algorithms we reduced them to $\sigma_k=0.25$ in the
second experiment, keeping all other parameters unchanged. A typical
scatterplot of such data is presented on the right panel of Fig. \ref
{FigScat12}.

The results of this experiment are presented in Table~\ref{TabEx2}.
They are compared graphically to the results of Experiment 1 in Fig.
\ref{FigVol12}. The covering frequencies are not significantly changed.
In comparison to Experiment 1, the average volumes decreased
significantly for EM-ellipsoids but not for the LS ones.
\end{experiment}

\begin{figure}[t]
\includegraphics{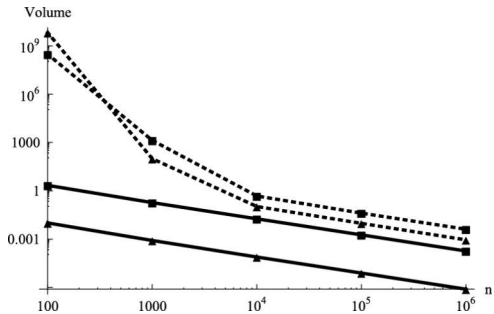}
\caption{Average volumes of ellipsoids in Experiment 1 ($\blacksquare
$) and Experiment 2 ($\blacktriangle$).
Solid lines for EM, dashed lines for LS (First component)}\label{FigVol12}
\vspace*{-6pt}
\end{figure}

\begin{table}[t!]
\caption{Experiment 2 results ($k$ is the number of component)}\label{TabEx2}
\begin{tabular}{|l|c|c|c|c|}
\multicolumn{5}{c}{Covering frequencies}\\
\hline
\tmultirow{2}{*}{$n$} &\multicolumn{2}{c|}{LS} & \multicolumn{2}{c|}{EM}\\
\cline{2-5}
& $k=1$ & $k=2$ & $k=1$ & $k=2$\\
\hline
100 & 0.886 & 0.922& 0.950& 0.943\\
$10^3$&0.942& 0.910& 0.945& 0.948\\
$10^4$&0.954& 0.946& 0.951& 0.955\\
$10^5$&0.962& 0.958& 0.943& 0.950\\
$10^6$&0.955& 0.937& 0.961& 0.942\\
\hline
\end{tabular}
\begin{tabular}{|l|c|c|c|c|}
\multicolumn{5}{c}{Average volume of ellipsoids}\\
\hline
\tmultirow{2}{*}{$n$} &\multicolumn{2}{c|}{LS} & \multicolumn{2}{c|}{EM}\\
\cline{2-5}
& $k=1$ & $k=2$ & $k=1$ & $k=2$\\
\hline
100    & 6560085466   & 43879747920  & 0.01022148               & 0.01199164               \\
$10^3$ & 98.92863     & 182799.39295 & 0.0008286214             & 0.0011644647             \\
$10^4$ & 0.1061491    & 0.2465760    & $7.846928\times10^{-05}$ & $1.196875\times10^{-04}$ \\
$10^5$ & 0.009584066  & 0.021603033  & $7.870776\times10^{-06}$ & $1.191609\times10^{-05}$ \\
$10^6$ & 0.0009045894 & 0.0021206426 & $7.875141\times10^{-07}$ & $1.189581\times10^{-06}$ \\
\hline
\end{tabular}
\end{table}
%

\begin{experiment}
Here we consider another set of parameters (see
Table~\ref{TabPar2}). The regression errors are Gaussian. In this model
the subjects cannot be classified uniquely by their observed variables
(see the left panel in Fig. \ref{FigScat34}).
\begin{figure}[t]
\includegraphics{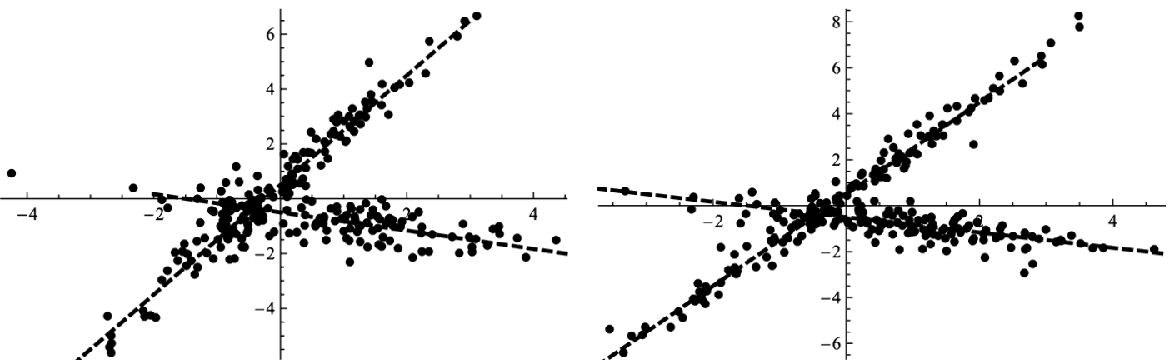}
\caption{Typical scatterplots of data in Experiment 3 (left) and
Experiment 4 (right)}\label{FigScat34}
\end{figure}

\begin{table}[t]
\caption{Parameters for simulation in Experiments 3 and 4}\label{TabPar2}
\begin{tabular}{ | l | c | r |}
\hline
$k$ & 1 & 2 \\
\hline
$\mu_k$ & 0 & 1 \\
\hline
$\varSigma_k$ & 2 & 2 \\
\hline
$\sigma_k$ & 0.5 & 0.5 \\
\hline
$b_0^k$ & 0.5 & $-$0.5 \\
\hline
$b_1^k$ & 2 & $-\frac{1}{3}$ \\
\hline
\end{tabular}
\end{table}

\begin{table}[t]
\caption{Experiment 3 results ($k$ is the number of component)}\label{TabEx3}
\begin{tabular}{|l|c|c|c|c|}
\multicolumn{5}{c}{Covering frequencies}\\
\hline
\tmultirow{2}{*}{$n$} &\multicolumn{2}{c|}{LS} & \multicolumn{2}{c|}{EM}\\
\cline{2-5}
& $k=1$ & $k=2$ & $k=1$ & $k=2$\\
\hline
100& 0.920& 0.928& 0.949& 0.935\\
$10^3$&0.953& 0.943& 0.948& 0.946\\
$10^4$&0.951& 0.957& 0.954& 0.945\\
$10^5$&0.947& 0.963& 0.942& 0.961\\
$10^6$&0.945& 0.951& 0.948& 0.939\\
\hline
\end{tabular}
\begin{tabular}{|l|c|c|c|c|}
\multicolumn{5}{c}{Average volume of ellipsoids}\\
\hline
\tmultirow{2}{*}{$n$} &\multicolumn{2}{c|}{LS} & \multicolumn{2}{c|}{EM}\\
\cline{2-5}
& $k=1$ & $k=2$ & $k=1$ & $k=2$\\
\hline
100& 294.5016& 28837.3340& 0.05897494& 0.06088698\\
$10^3$&0.6088472& 0.6274452& 0.005250821& 0.004937218\\
$10^4$&0.05837274& 0.05594969& 0.0005024635& 0.0004993278\\
$10^5$&0.005604424& 0.00551257& 4.987135$\times10^{-05}$&
5.024126$\times10^{-05}$\\
$10^6$&0.0005625693&0.0005550716 &4.978973$\times10^{-06}$&
5.029275$\times10^{-06}$\\
\hline
\end{tabular}
\vspace*{6pt}
\end{table}
The results are presented in Table~\ref{TabEx3}. Again, the
EM-ellipsoids outperform the LS ones.
\end{experiment}

\begin{experiment}
In this experiment the parameters are the same as in Experiment 3, but
the regression errors are \textbf{not} Gaussian. We let
$\varepsilon^k=\sqrt{3/5}\sigma_k\eta$, where $\eta$ has the Student-T
distribution with 5 degrees of freedom. So the errors here have the
same variances as in Experiment 3, but their distributions are
heavy-tailed. Note that 5 is the minimal number of degrees of freedom
for which the assumption $\M(\varepsilon_k)^4$ of Theorem~\ref{ThCLT} holds.

A typical data scatterplot for this model is presented on the right
panel of Fig. \ref{FigScat34}. It is visually indistinguishable from
the typical pattern of the Gaussian model from Experiment 3, presented
on the left panel.

Results of this experiment are presented in Table~\ref{TabEx4}. Note
that in this case the covering proportion of the EM-ellipsoids does not
tend to the nominal $1-\alpha=0.95$ for large $n$. The covering
proportion of LS-ellipsoids is much nearer to 0.95. So the heavy tails
of distributions of the regression errors deteriorate performance of
(Gaussian model based) EM-ellipsoids but not of nonparametric LS-ellipsoids.\vadjust{\goodbreak}
\end{experiment}

\section{An application to sociological data analysis}\label{SecApplication}
To demonstrate possibilities of the developed technique, we present a
toy example of construction of confidence ellipsoids in statistical
analysis of dependence between school
performance of students and political attitudes of their adult
environment. The analysis was based on two data sets.
The first one contains results of the External independent testing in
Ukraine in 2016 -- EIT-2016. EIT is a a set of exams for high schools
graduates for admission to universities. Data on EIT-2016\footnote
{Taken from the official site of \textit{Ukrainian Center for
Educational Quality Assessment} \surl{https://zno.\\testportal.com.ua/stat/2016}.} contain individual scores
of examinees with some additional
information including the region of Ukraine at which the examinee's
school was located.
The scores range from 100 to 200 points.

\begin{table}
\caption{Experiment 4 results ($k$ is the number of component)}\label{TabEx4}
\begin{tabular}{|l|c|c|c|c|}
\multicolumn{5}{c}{Covering frequencies}\\
\hline
\tmultirow{2}{*}{$n$} &\multicolumn{2}{c|}{LS} & \multicolumn{2}{c|}{EM}\\
\cline{2-5}
& $k=1$ & $k=2$ & $k=1$ & $k=2$\\
\hline
100& 0.912& 0.915& 0.943& 0.937\\
$10^{3}$& 0.948& 0.945& 0.949& 0.959\\
$10^{4}$& 0.937& 0.945& 0.929& 0.953\\
$10^{5}$& 0.947& 0.951& 0.915& 0.930\\
$10^{6}$& 0.961& 0.953& 0.634& 0.763\\
\hline
\end{tabular}
\begin{tabular}{|l|c|c|c|c|}
\multicolumn{5}{c}{Average volume of ellipsoids}\\
\hline
\tmultirow{2}{*}{$n$} &\multicolumn{2}{c|}{LS} & \multicolumn{2}{c|}{EM}\\
\cline{2-5}
& $k=1$ & $k=2$ & $k=1$ & $k=2$\\
\hline
100      & 997.1288     & 584.8507     & 0.06740671                     & 0.06419959                     \\
$10^{3}$ & 0.7006367    & 0.6127510    & 0.005262779                    & 0.004971307                    \\
$10^{4}$ & 0.05798962   & 0.05624429   & 0.0004850319                   & 0.0004884329                   \\
$10^{5}$ & 0.005621060  & 0.005574176  & $4.667732\times10^{-05}$ & $4.746252\times10^{-05}$ \\
$10^{6}$ & 0.0005616700 & 0.0005566846 & $4.666926\times10^{-06}$ & $4.745760\times10^{-06}$ \\
\hline
\end{tabular}
\end{table}

We considered the information on the scores
on two subjects: \textit{Ukrainian language and literature} (Ukr) and
on \textit{Mathematics} (Math).
EIT-2016 contains data on these scores for nearly 246 000 examinees.
It is obvious that Ukr and Math scores should be dependent and the
simplest way to model this dependency is the linear regression:
\[
\text{Ukr}=b_0+b_1\text{Mat}+\varepsilon.
\]
We suppose that the coefficients $b_0$ and $b_1$ may depend on the
political attitudes of the adult environment in which the student was
brought up. Say, in a family of Ukrainian independence adherents one
expects more interest to Ukrainian language than in an environment
critical toward the Ukrainian state existence.\vadjust{\eject}

Of course EIT-2016 does not contain any information on political
issues. So we used the second data set with the official data on the
results\footnote{See the site of \textit{Central Election Commission
(Ukraine)} \url{http://www.cvk.gov.ua/vnd\_2014/}.} of the
Ukrainian Parliament elections-2014 to get approximate proportions of
adherents of different political choices in different regions of Ukraine.

29 political parties and blocks took part in the elections. The voters
were able also to vote against all or not to take part in the voting.
We divided all the population of voters into three components:

(1) Persons who voted for parties which then created the ruling
coalition (BPP, People's front, Fatherland, Radical party, Self
Reliance). This is the component of persons with positive attitudes to
the pro-European Ukrainian state.

(2) Persons who voted for the Opposition block, voters against all,
and voters for small parties which where under 5\% threshold at these
elections. These are voters critical to the pro-European line of
Ukraine but taking part in the political life of the state.

(3) Persons who did not take part in the voting. These are persons who
did not consider Ukrainian state as their own one or are not interested
in politics at all.

We used the results of elections to calculate the proportions of each
component in each region of Ukraine where the voting was held. These
proportions were taken as estimates for the probabilities that a
student from a corresponding region was brought up in the environment
of a corresponding component. That is, they were considered as the
mixing probabilities.

The LS- and EM-ellipsoids for $b_0$ and $b_1$ obtained by these data
are presented on Fig. \ref{FigEIT}. The ellipsoids were constructed
with the significance level $\alpha=0.05/3\approx0.01667$, so by the
Bonferroni rule, they are unilateral confidence sets with $\alpha
=0.05$. Since the ellipsoids are not intersecting in both cases, one
concludes that the vectors of regression coefficients $(b_0^i,b_1^i)$,
$i=1,\dots,3$ are significantly different for different components.

Note that the EM approach leads to estimates significantly different
from the LS ones. This may suggest that the normal mixture model
(\ref{EqNormalMix}) does not hold for the data. Does the nonparametric
model hold for them? Analysis of this problem and meaningful
sociological interpretation of these results lie beyond the scope of
this article.

\begin{figure}[t]
\includegraphics{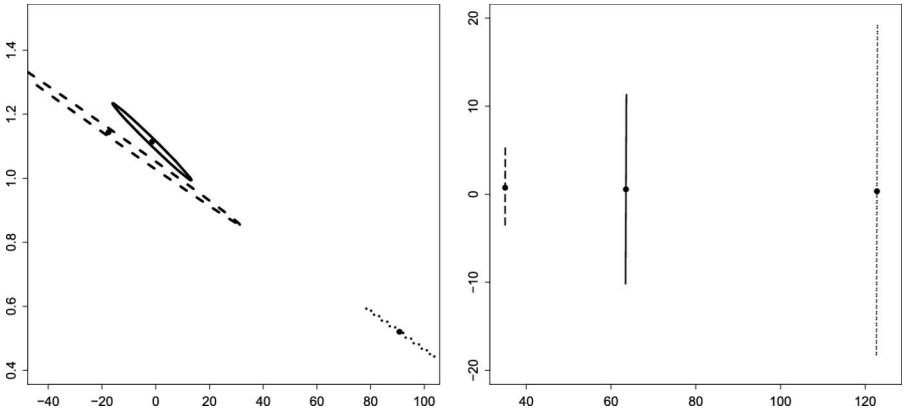}
\caption{LS (left) and EM (right) confidence ellipsoids for the EIT data.
Components: (1) dotted line, (2) dashed line, (3) solid line.
$b_0$ on the horizontal axis, $b_1$ on the vertical axis}\label{FigEIT}
\end{figure}

\section{Concluding remarks}\label{SecConcl}

We considered two approaches to the construction of confidence sets for
coefficients of the linear regression in the mixture model with varying
mixing probabilities. Both approaches demonstrate sufficient agreement
of nominal and real significance levels for sufficiently large samples
when the data satisfy underlying assumptions of the confidence set
construction technique. The parametric approach needs a
significant additional a priori information in comparison with the
nonparametric one. But it utilizes this information providing much
smaller confidence sets than in the nonparametric case.

On the other hand, the nonparametric estimators proved to be a good
initial approximation for the construction of parametric estimators via
the EM-algorithm. Nonparametric confidence sets also perform adequately
in the cases when the assumptions of parametric model are broken.


\begin{acknowledgement}[title={Acknowledgments}]
We are thankful to the unknown referees for their attention to our work
and fruitful comments.

The research was supported in part by the Taras Shevchenko National
University of Kyiv scientific grant
N 16\textcyr{B}\textcyr{F}038-02.
\end{acknowledgement}




\begin{thebibliography}{16}

\bibitem{Autin:TestDensities}
\begin{barticle}
\bauthor{\bsnm{Autin}, \binits{F.}},
\bauthor{\bsnm{Pouet}, \binits{Ch.}}:
\batitle{Test on the components of mixture densities}.
\bjtitle{Statistics \& Risk Modeling}
\bvolume{28}(\bissue{4}),
\bfpage{389}--\blpage{410}
(\byear{2011}).
\bid{doi={10.1524/strm.2011.1065}, mr={2877572}}
\end{barticle}
%
\OrigBibText
Autin, F., Pouet, Ch.: Test on the components of mixture
densities. \textit{Statistics \& Risk Modeling} \textbf{28}, No 4,
389--410 (2011)
\endOrigBibText
\bptok{structpyb}
\endbibitem

\bibitem{Benaglia}
\begin{barticle}
\bauthor{\bsnm{Benaglia}, \binits{T.}},
\bauthor{\bsnm{Chauveau}, \binits{D.}},
\bauthor{\bsnm{Hunter}, \binits{D.R.}},
\bauthor{\bsnm{mixtools}, \binits{Y.D.S.}}:
\batitle{An R Package for Analyzing Finite Mixture Models}.
\bjtitle{Journal of Statistical Software}
\bvolume{32}(\bissue{6}),
\bfpage{1}--\blpage{29}
(\byear{2009}).
\bid{doi={\\10.18637/jss.v032.i06}}
\end{barticle}
%
\OrigBibText
Benaglia T., Chauveau D., Hunter D. R., and Young D. S. \texttt
{mixtools}: An R Package for Analyzing Finite
Mixture Models. \textit{Journal of Statistical Software}, \textbf{32}(6):1--29 (2009).
\endOrigBibText
\bptok{structpyb}
\endbibitem

\bibitem{BorovkovMs}
\begin{bbook}
\bauthor{\bsnm{Borovkov}, \binits{A.A.}}:
\bbtitle{Mathematical statistics}.
\bpublisher{Gordon and Breach Science Publishers},
\blocation{Amsterdam}
(\byear{1998}).
\bid{mr={1712750}}
\end{bbook}
%
\OrigBibText
Borovkov, A.A.: Mathematical statistics. Gordon and Breach Science Publishers,
Amsterdam (1998)
\endOrigBibText
\bptok{structpyb}
\endbibitem

\bibitem{Faria}
\begin{barticle}
\bauthor{\bsnm{Faria}, \binits{S.}}:
\batitle{Soromenhob Gilda Fitting mixtures of linear regressions}.
\bjtitle{Journal of Statistical Computation and Simulation}
\bvolume{80}(\bissue{2}),
\bfpage{201}--\blpage{225}
(\byear{2010}).
\bid{doi={\\10.1080/00949650802590261}, mr={2757044}}
\end{barticle}
%
\OrigBibText
Faria Susana, Soromenhob Gilda Fitting mixtures of linear regressions, \textit{Journal
of Statistical Computation and Simulation} \textbf{80}, No 2, 201--225 (2010)
\endOrigBibText
\bptok{structpyb}
\endbibitem

\bibitem{GrunLeisch}
\begin{bchapter}
\bauthor{\bsnm{Gr\"{u}n}, \binits{B.}},
\bauthor{\bsnm{Friedrich}, \binits{L.}}:
\bctitle{Fitting finite mixtures of linear regression models with varying \& fixed effects in R}.
In: \beditor{\bsnm{Rizzi}, \binits{A.}},
\beditor{\bsnm{Vichi}, \binits{M.}} (eds.)
\bbtitle{Compstat 2006 -- Proceedings in Computational Statistics},
pp.~\bfpage{853}--\blpage{860}.
\bpublisher{Physica Verlag},
\blocation{Heidelberg, Germany}
(\byear{2006}).
\bid{mr={2173118}}
\end{bchapter}
%
\OrigBibText
Gr\"un Bettina and Leisch Friedrich: Fitting finite mixtures of linear
regression models
with varying \& fixed effects in R.
In Alfredo Rizzi and Maurizio
Vichi, editors,
\textit{Compstat 2006 - Proceedings in Computational Statistics}, 853-860.
Physica Verlag, Heidelberg, Germany, (2006).
\endOrigBibText
\bptok{structpyb}
\endbibitem

\bibitem{Liubashenko}
\begin{barticle}
\bauthor{\bsnm{Liubashenko}, \binits{D.}},
\bauthor{\bsnm{Maiboroda}, \binits{R.}}:
\batitle{Linear regression by observations from mixture with varying concentrations}.
\bjtitle{Modern Stochastics: Theory and Applications}
\bvolume{2}(\bissue{4}),
\bfpage{343}--\blpage{353}
(\byear{2015}).
\bid{doi={10.15559/15-VMSTA41}, mr={3456142}}
\end{barticle}
%
\OrigBibText
Liubashenko D., Maiboroda R. Linear regression by observations from mixture
with varying concentrations. \textit{Modern Stochastics: Theory and Applications},
\textbf{2}, No 4, 343 - 353, (2015)
\endOrigBibText
\bptok{structpyb}
\endbibitem

\bibitem{Maiboroda2003}
\begin{bbook}
\bauthor{\bsnm{Maiboroda}, \binits{R.}}:
\bbtitle{Statistical analysis of mixtures}.
\bpublisher{Kyiv University Publishers},
\blocation{Kyiv}
(\byear{2003}).
\bcomment{(in Ukrainian)}
\end{bbook}
%
\OrigBibText
Maiboroda, R.: Statistical analysis of mixtures. Kyiv
University Publishers, Kyiv (in Ukrainian) (2003)
\endOrigBibText
\bptok{structpyb}
\endbibitem

\bibitem{Kub1}
\begin{barticle}
\bauthor{\bsnm{Maiboroda}, \binits{R.}},
\bauthor{\bsnm{Kubaichuk}, \binits{O.}}:
\batitle{Asymptotic normality of improved weighted empirical distribution functions}.
\bjtitle{Theor. Probab. Math. Stat.}
\bvolume{69},
\bfpage{95}--\blpage{102}
(\byear{2004}).
\bid{doi={10.1090/S0094-9000-05-00617-4}, mr={2110908}}
\end{barticle}
%
\OrigBibText
Maiboroda R., Kubaichuk O., Asymptotic normality of improved weighted
empirical distribution functions.
\textit{Theor. Probab. Math. Stat.} \textbf{69}, 95-102, (2004).
\endOrigBibText
\bptok{structpyb}
\endbibitem

\bibitem{MS2008}
\begin{bbook}
\bauthor{\bsnm{Maiboroda}, \binits{R.E.}},
\bauthor{\bsnm{Sugakova}, \binits{O.V.}}:
\bbtitle{Estimation and classification by observations from mixture}.
\bpublisher{Kuiv University Publishers},
\blocation{Kyiv}
(\byear{2008}).
\bcomment{(in Ukrainian)}
\end{bbook}
%
\OrigBibText
Maiboroda R.E., Sugakova O.V. Estimation and classification by
observations from mixture. Kuiv
University Publishers, Kyiv (in Ukrainian) (2008)
\endOrigBibText
\bptok{structpyb}
\endbibitem

\bibitem{Maiboroda:StatisticsDNA}
\begin{barticle}
\bauthor{\bsnm{Maiboroda}, \binits{R.}},
\bauthor{\bsnm{Sugakova}, \binits{O.}}:
\batitle{Statistics of mixtures with varying concentrations with application to DNA microarray data analysis}.
\bjtitle{Journal of nonparametric statistics.}
\bvolume{24}(\bissue{1}),
\bfpage{201}--\blpage{205}
(\byear{2012}).
\bid{doi={10.1080/10485252.2011.630076}, mr={2885834}}
\end{barticle}
%
\OrigBibText
Maiboroda, R., Sugakova, O.: Statistics of mixtures with varying
concentrations with application to DNA microarray data analysis.
\textit{Journal of nonparametric statistics}. \textbf{24} , No 1 201--205
(2012)
\endOrigBibText
\bptok{structpyb}
\endbibitem

\bibitem{Maiboroda:AdaptMVC}
\begin{barticle}
\bauthor{\bsnm{Maiboroda}, \binits{R.E.}},
\bauthor{\bsnm{Sugakova}, \binits{O.V.}},
\bauthor{\bsnm{Doronin}, \binits{A.V.}}:
\batitle{Generalized estimating equations for mixtures with varying concentrations}.
\bjtitle{The Canadian Journal of Statistics}
\bvolume{41}(\bissue{2}),
\bfpage{217}--\blpage{236}
(\byear{2013}).
\bid{doi={10.1002/cjs.11170}, mr={3061876}}
\end{barticle}
%
\OrigBibText
Maiboroda, R. E., Sugakova, O. V., Doronin, A. V.: Generalized
estimating equations for mixtures with varying concentrations. \textit{The
Canadian Journal of Statistics} \textbf{41} , No 2, 217--236 (2013)
\endOrigBibText
\bptok{structpyb}
\endbibitem

\bibitem{MS2015-2}
\begin{barticle}
\bauthor{\bsnm{Maiboroda}, \binits{R.}},
\bauthor{\bsnm{Sugakova}, \binits{O.}}:
\batitle{Sampling bias correction in the model of mixtures with varying concentrations}.
\bjtitle{Methodology and Computing in Applied Probability}.
\bvolume{17}(\bissue{1})
(\byear{2015}).
\bid{doi={10.1007/s11009-013-9349-4}, mr={3306681}}
\end{barticle}
%
\OrigBibText
Maiboroda, R. Sugakova O.: Sampling bias correction in the model of mixtures
with varying concentrations. \textit{Methodology and Computing in Applied Probability}.
\textbf{17} No 1 (2015)
\endOrigBibText
\bptok{structpyb}
\endbibitem

\bibitem{McLachlan}
\begin{bbook}
\bauthor{\bsnm{McLachlan}, \binits{G.}}:
\bbtitle{Krishnan Thriyambakam The EM Algorithm and Extensions},
\bedition{2}nd edn.
\bpublisher{Wiley},
(\byear{2008}).
\bid{doi={10.1002/9780470191613}, mr={2392878}}
\end{bbook}
%
\OrigBibText
McLachlan Geoffrey , Krishnan Thriyambakam The EM Algorithm and
Extensions, 2nd Edition Wiley, (2008).
\endOrigBibText
\bptok{structpyb}
\endbibitem

\bibitem{Seber}
\begin{bbook}
\bauthor{\bsnm{Seber}, \binits{G.A.F.}},
\bauthor{\bsnm{Lee}, \binits{A.J.}}:
\bbtitle{Linear Regression Analysys}.
\bpublisher{Wiley},
(\byear{2003}).
\bid{doi={10.1002/9780471722199}, mr={1958247}}
\end{bbook}
%
\OrigBibText
Seber G.A.F., Lee A.J.: Linear Regression Analysys.
Wiley (2003)
\endOrigBibText
\bptok{structpyb}
\endbibitem

\bibitem{Shao}
\begin{bbook}
\bauthor{\bsnm{Shao}, \binits{J.}}:
\bbtitle{Mathematical statistics}.
\bpublisher{Springer},
\blocation{New York}
(\byear{1998}).
\bid{mr={1670883}}
\end{bbook}
%
\OrigBibText
Shao J.: Mathematical statistics. Springer-Verlag: New York,
(1998)
\endOrigBibText
\bptok{structpyb}
\endbibitem

\bibitem{Titterington}
\begin{bbook}
\bauthor{\bsnm{Titterington}, \binits{D.M.}},
\bauthor{\bsnm{Smith}, \binits{A.F.}},
\bauthor{\bsnm{Makov}, \binits{U.E.}}:
\bbtitle{Analysis of Finite Mixture Distributions}.
\bpublisher{Wiley},
\blocation{New York}
(\byear{1985}).
\bid{mr={0838090}}
\end{bbook}
%
\OrigBibText
D. M. Titterington, A. F. Smith, U. E. Makov, Analysis of Finite
Mixture Distributions. Wiley, New York (1985)
\endOrigBibText
\bptok{structpyb}
\endbibitem

\end{thebibliography}
\end{document}